\documentclass[envcountsame]{llncs}
\usepackage{jeffe,graphicx,hyperref}
\usepackage[utf8]{inputenc}
\usepackage{marvosym}
\usepackage{verbatim}
\usepackage{xspace}
\usepackage{subfigure}

\usepackage{microtype}

\usepackage[mathcal]{euscript}
\usepackage{stmaryrd}
\usepackage{paralist}

\def\Real{\mathbb{R}}

\def\GreedyFuture{\mathsc{GreedyFuture}\xspace}
\def\Patrascu{P\u{a}tra\c{s}cu\xspace}
\def\OPT{\mathsc{OPT}\xspace}
\def\GreedyASS{\mathsc{GreedyASS}\xspace}

\let\path p

\begin{document}

\title{Upper Bounds for Maximally Greedy Binary Search Trees}

\author{Kyle Fox}

\institute{Department of Computer Science, University of Illinois,
  Urbana-Champaign \email{kylefox2@illinois.edu}}

\maketitle

\begin{abstract}
  At SODA 2009, Demaine et~al. presented a novel connection between binary
  search trees (BSTs) and subsets of points on the plane. This connection
  was independently discovered by Derryberry et al.
  As part of their
  results, Demaine et~al.
  considered \GreedyFuture, an offline BST algorithm that
  greedily rearranges the search path to minimize the cost of future searches.
  They showed that \GreedyFuture is actually an online algorithm in their
  geometric view, and that there is a way to turn \GreedyFuture into an online
  BST algorithm with
  only a constant factor increase in total search cost. Demaine et~al.
  conjectured this algorithm was dynamically optimal, but no upper bounds
  were given in their paper.
  We prove the first non-trivial
  upper bounds for the cost of search operations using
  \GreedyFuture including giving an access lemma similar to that found in
  Sleator and Tarjan's classic paper on splay trees.
\end{abstract}

\section{Introduction}

The \EMPH{dynamic optimality conjecture} states that 
given a sequence of successful searches on
an $n$-node binary
search tree, the number of nodes accessed by splay trees is at
most a constant times the number of node accesses and rotations performed
by the optimal algorithm
\emph{for that sequence}. Sleator and Tarjan gave this
conjecture in their paper on splay trees in which they showed~$O(\log n)$
amortized performance as well as several other upper bounds~\cite{st-sbst-85}.
Proving the dynamic optimality conjecture
seems very difficult. There is no known polynomial time algorithm for
finding an optimal BST in the offline setting where we know all searches
in advance,\footnote{In fact, the exact optimization
problem becomes NP-hard if we must access an
arbitrary number of specified nodes during each search~\cite{dhikp-gbst-09}.}
and this conjecture states
that splaying is a simple solution to the online problem.

Until recently, there has been little progress made directly related to this
conjecture. Wilber gave two lower bounds on the number of accesses needed
for any given search sequence~\cite{w-lbabs-89}. There are a handful
of online BST algorithms that are $O(\log \log n)$-competitive
\cite{dhip-doa-07,cds-lcbst-06,c-mt-06,bddf-lcbst-10}, but no
upper bound is known
for the competitiveness of splay trees except the trivial~$O(\log n)$.

\subsection{A Geometric View}

Recently, Demaine et~al. introduced a new way of conceptualizing BSTs
using geometry~\cite{dhikp-gbst-09}. A variant of
this model was independently discovered
by Derryberry et al.~\cite{dsw-lbfbs-05}.
In the geometric view, BST node
accesses are represented as points~$(x,y)$ where~$x$ denotes the rank of the
accessed node and~$y$ represents which search accessed the node. A pair of
points~$a$ and~$b$ in point set~$P$
are called \emph{arborally satisfied} if they lie on the
same horizontal or vertical line, or if the
closed rectangle with corners~$a$ and~$b$
contains another point from~$P$. The family of arborally satisfied
point sets corresponds exactly to BST accesses when rotations upon accessed
nodes are allowed~\cite{dhikp-gbst-09}.

By starting with a point set~$X$ that represents the points a BST must access
to complete searches in a given search sequence~$S$,
we can describe an optimal BST algorithm for~$S$ as
a minimum superset of~$X$ that is arborally satisfied~\cite{dhikp-gbst-09}.
This correspondence between BSTs and arborally satisfied supersets
allows us to focus on algorithms strictly in the geometric
view. Additionally, it is possible to show lower bounds for the BST model
by showing the same for the geometric model. Demaine et~al. take advantage
of this fact to show a class of lower bounds that supersede the lower bounds
of Wilber~\cite{w-lbabs-89,dhikp-gbst-09}.
Further, it is possible to describe an
\emph{online} version of the arborally satisfied superset problem and transform
solutions to this problem into online BST algorithms with at most a constant
factor increase in cost~\cite{dhikp-gbst-09}.

\subsection{Being Greedy}

Demaine et~al. also consider an offline BST algorithm they call \GreedyFuture,
originally proposed by Lucas~\cite{l-cfcbs-88} and
Munro~\cite{m-ocls-00}. \GreedyFuture only touches nodes on the
search path, and then rearranges the search path in order to greedily
minimize the time for upcoming searches.

The worst-case example known for the competitiveness
of \GreedyFuture is a complete
binary search tree with searches performed in bit-reversal order upon the
leaves~\cite{m-ocls-00}. \GreedyFuture has an amortized cost of~$\lg n$
per search
on this sequence. The optimal algorithm rotates the leaves closer to the
root and obtains an amortized cost of~$\lg \frac{n}{2} + o(1)$.
Given a search sequence
of length~$m$, let~$\OPT$ be the total cost of the optimal
algorithm for that sequence.
Demaine et al. conjecture that \GreedyFuture is $O(1)$-competitive.
In fact, the bit-reversal example
suggests that the cost of \GreedyFuture is at most $\OPT + m$;
it appears optimal within an \emph{additive term}.

Surprisingly, Demaine et~al. showed that \GreedyFuture, an offline algorithm
that uses very strong knowledge about the future, is actually an online
algorithm in the geometric model~\cite{dhikp-gbst-09}.
Recall that online algorithms in the geometric model correspond to 
online algorithms in the BST model with essentially the same cost.
If \GreedyFuture is actually an offline
dynamically optimal BST algorithm as it appears to
be, then there exists an \emph{online}
dynamically optimal BST algorithm.

\subsection{Our Contributions}

Despite the apparent optimality of the \GreedyFuture algorithm, nothing was
known about its amortized behavior when Demaine et~al. wrote their report.
We provide the first theoretical evidence that \GreedyFuture is an optimal
algorithm in the following forms:
\begin{compactitem}
\item
  An access lemma similar to that used by Sleator and Tarjan
  for splay trees~\cite{st-sbst-85}.
  This lemma implies several upper bounds including~$O(\log n)$ amortized
  performance.
\item
  A sequential access theorem that states \GreedyFuture takes linear time
  to access all nodes in order starting from any arbitrary BST.
\end{compactitem}

We heavily use the geometric model of Demaine et~al. to prove
the access lemma while focusing directly on BSTs to prove the sequential
access theorem. It is our hope that these results will create further interest
in studying \GreedyFuture as its structural properties seem well suited
for further theoretical analysis (the proof of the sequential access theorem
takes only a page). Additionally, the proof of the access lemma may provide
additional insight into other algorithms running in the geometric model.

\subsection{A Note on Independent Work}

John Iacono and Mihai \Patrascu have discovered a similar access lemma to that
given here using different proof techniques from those shown below.
The author learned about their work
via personal correspondence with them and Erik
Demaine well into performing the research contained in this report.
Their results have never been published.

Additionally, the author became aware of work by Goyal and
Gupta~\cite{gg-dobst-11} after initially writing this report.
They show \GreedyFuture has~$O(\log n)$ amortized performance.
This result appears in our paper as Corollary~\ref{cor:balance}.
As in our proof, they use the geometric model, but they do not use a potential
function as we do to prove a more general access lemma.

\section{Arboral and Geometric Models of BSTs}
\subsection{The Arboral Model}

We will consider the same BST model used by
Demaine et al.~\cite{dhikp-gbst-09}. We consider
only successful searches and not insertions or deletions. Let~$n$
and~$m$ be the number of elements in the search tree and the number of searches
respectively. We assume the elements have distinct keys in~$\Set{1,\dots,n}$.

Given a BST~$T_1$, a subtree~$\tau$ of~$T_1$ containing the
root, and a tree~$\tau'$ on the same nodes as~$\tau$, we say~$T_1$ can be
\EMPH{reconfigured}
by an operation $\tau\to\tau'$ to another BST~$T_2$ if~$T_2$ is
identical to~$T_1$ except for~$\tau$ being replaced by~$\tau'$. The cost of
the reconfiguration is $|\tau|=|\tau'|$.

Given a search sequence $S = \langle s_1,s_2,\dots,s_m\rangle$, we say
a BST algorithm \EMPH{executes}~$S$ by an execution
$E = \langle T_0,\tau_1\to\tau'_1,\dots,\tau_m\to\tau'_m\rangle$ if all
reconfigurations are performed on
subtrees containing the root, and $s_i \in \tau_i$ for all~$i$.
For~$i = 1,2,\dots, m$, define~$T_i$ to be~$T_{i-1}$ with the reconfiguration
$\tau_i\to\tau'_i$. The cost of execution~$E$ is $\sum_{i=1}^m|\tau_i|$.

As explained by Demaine et al.~\cite{dhikp-gbst-09},
this model is constant-factor equivalent
to other reasonable BST models such as those by Wilber and
Lucas~\cite{w-lbabs-89,l-cfcbs-88}.

\subsection{The Geometric Model}

We now turn our focus to the geometric model as given by Demaine
et al. \cite{dhikp-gbst-09}.
Define a \EMPH{point}~$p$ to be a point in 2D with
integer coordinates~$(p.x,p.y)$ such
that $1\leq p.x\leq n$ and $1 \leq p.y \leq m$. Let~$\square ab$ denote the
closed axis-aligned rectangle with corners~$a$ and~$b$.

A pair of points~$(a,b)$ (or their induced rectangle~$\square ab$) is
\EMPH{arborally satisfied}
with respect to a point set~$P$ if (1)~$a$ and~$b$ are
orthogonally collinear (horizontally or vertically aligned), or (2) there is
at least one point from ~$P \setminus \Set{a,b}$ in~$\square ab$. A point
set~$P$ is arborally satisfied if all pairs of points in~$P$ are arborally
satisfied with respect to~$P$. See Fig.~\ref{fig:satdef-unsat} and
Fig.~\ref{fig:satdef-sat}.

As explained in~\cite{dhikp-gbst-09}, there is a one-to-one correspondence
between BST executions and arborally satisfied sets of points.
Let the \EMPH{geometric view} of a BST execution~$E$ be the point set
${P(E) = \Set{(x,y) | x \in \tau_y}}$.
The point set~$P(E)$ for any BST execution~$E$ is arborally
satisfied~\cite{dhikp-gbst-09}.
Further,
for any arborally satisfied point set~$X$, there exists a BST execution~$E$
with $P(E)=X$~\cite{dhikp-gbst-09}.

Let the \EMPH{geometric view} of an access sequence~$S$ be the set of points
$P(S)=\Set{(s_1,1),(s_2,2),\dots,(s_m,m)}$. The above facts suggest that
finding an optimal BST algorithm for~$S$ is equivalent to finding a
minimum cardinality arborally satisfied superset of~$S$.
Due to this equivalence with BSTs, we will refer to values
in~$\Set{1,\dots,n}$ as \EMPH{elements}.

Naturally, we may want to use
the geometric model to find dynamically optimal \emph{online} BST algorithms.
The \EMPH{online arborally satisfied superset}
(online ASS) problem is to design an
algorithm that receives a sequence of points
$\langle(s_1,1),(s_2,2),\dots,(s_m,m)\rangle$
incrementally. After receiving the~$i$th point~$(s_i,i)$,
the algorithm must output a
set~$P_i$ of points on the line~$y=i$ such that\\
${\Set{(s_1,1),(s_2,2),\dots,(s_i,i)}\cup P_1\cup P_2\cup\cdots\cup P_i}$ is
arborally satisfied. The cost of the algorithm is~$m+\sum_{i=1}^m|P_i|$.

We say an online ASS algorithm performs a \EMPH{search} at time~$i$ when
it outputs the set~$P_i$. Further,
we say an online ASS algorithm \EMPH{accesses}~$x$ at time~$i$ if~$(x,i)$
is included in the input set of points or in~$P_i$.
The (non-amortized) cost of a search at time~$i$ is $|P_i| + 1$.

Unfortunately, the algorithm used to create a BST execution from an
arborally satisfied point set requires
knowledge about points above the line~$y=i$ to construct~$T_i$
\cite{dhikp-gbst-09}. We are not able
to go directly from a solution to the online ASS
problem to a solution for the online BST problem with
exactly the same cost. However,
this transformation is possible if we allow the cost of the BST algorithm
to be at most a constant multiple of the ASS
algorithm's cost~\cite{dhikp-gbst-09}.

\begin{figure}[t]
  \centering
  \begin{minipage}[t]{0.45\linewidth}
    \centering
    \includegraphics{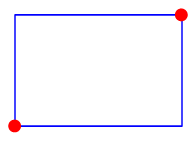}
    \caption{An unsatisfied pair of points. The closed axis-aligned rectangle
      with corners defined by the pair is shown.}
    \label{fig:satdef-unsat}
  \end{minipage}
  \hspace{0.5cm}
  \begin{minipage}[t]{0.45\linewidth}
    \centering
    \includegraphics{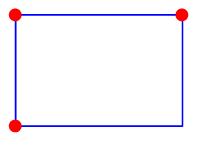}
    \caption{An arborally satisfied superset of the same pair of points}
    \label{fig:satdef-sat}
  \end{minipage}
\end{figure}

\section{\GreedyFuture}

We now turn our focus to describing the \GreedyFuture algorithm in more detail.
Let $S=\langle s_1,\dots,s_m\rangle$ be an arbitrary search sequence of
length~$m$. After every search, \GreedyFuture will rearrange the search path
to minimize the cost of future searches.

More precisely, consider the~$i$th search for the given sequence~$S$. If~$i=m$,
then \GreedyFuture does not rearrange the search path. Otherwise, if~$s_{i+1}$
lies on the search path~$\tau_i$, \GreedyFuture makes~$s_{i+1}$
the root of~$\tau'_i$. If~$s_{i+1}$
does not lie along the search path, then \GreedyFuture makes the
predecessor and successor of~$s_{i+1}$ \emph{within}~$\tau_i$ the root and
root's right child of~$\tau'_i$ (if the successor (predecessor) does not exist,
then \GreedyFuture makes the predecessor (successor) the root and does not
assign a right (left) child within~$\tau'_i$.) Now that it has fixed one or two 
nodes~$x_{\ell}$ and~$x_r$ with $x_{\ell}<x_r$,
\GreedyFuture recursively sets the remaining nodes of~$\tau_i$ \emph{less}
than~$x_{\ell}$ using the subsequence
of $\langle s_{i+1},\dots,s_m\rangle$ containing
nodes less than~$x_{\ell}$. It
then sets the nodes of~$\tau_i$ greater than~$x_r$ using the subsequence
of $\langle s_{i+1},\dots,s_m\rangle$ containing
nodes greater than~$x_r$.

Taking a cue from Demaine et~al., we will call the
online geometric model of the algorithm
\GreedyASS. Let $X = P(S)$ for some BST access sequence~$S$. At each time~$i$,
\GreedyASS simply outputs the minimal set of points at~$y=i$ needed to
satisfy~$X$ up to $y\leq i$.

We note that the set of points needed to satisfy~$X$ up to~$y\leq i$ is uniquely
defined. For each unsatisfied rectangle formed with~$(s_i,i)$ in one corner,
we add the other corner at $y=i$.
We can also define \GreedyASS as an algorithm that sweeps right and left
from the search node, accessing nodes that have increasingly greater last
access times. See Fig.~\ref{fig:greedyass}.

\GreedyASS, the online geometric view of \GreedyFuture, greatly reduces the
complexity of predicting \GreedyFuture's behavior. By focusing our attention
on this geometric algorithm, we proceed to prove several upper bounds
on both algorithms' performance in the following section.

\begin{figure}[t]
  \centering
  \begin{minipage}[t]{0.45\linewidth}
    \centering
    \includegraphics{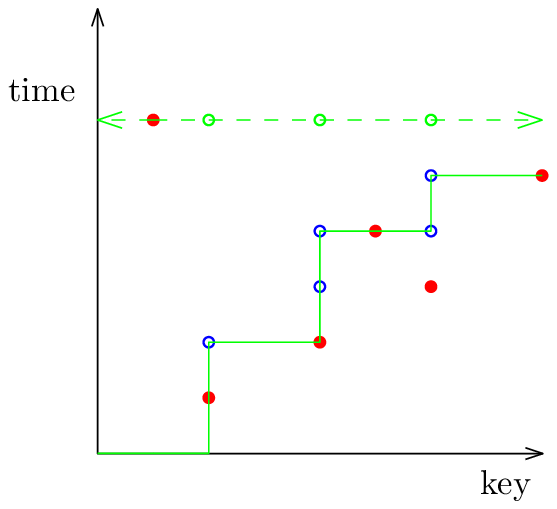}
  \end{minipage}
  \hspace{0.5cm}
  \begin{minipage}[t]{0.45\linewidth}
    \centering
    \includegraphics{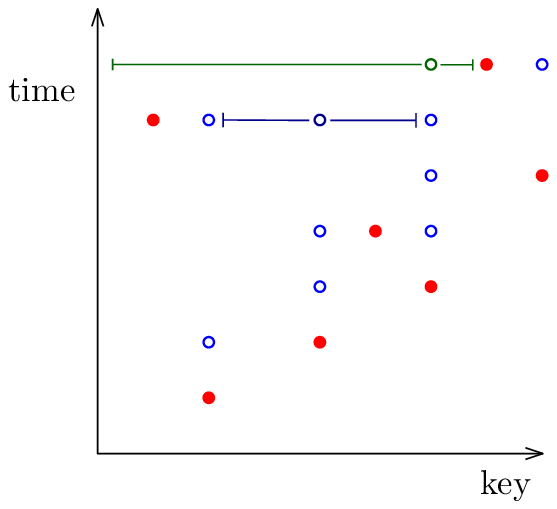}
  \end{minipage}
  \caption{(Left) A sample execution of \GreedyASS.
    Search elements are represented as solid disks.
    For the latest search, \GreedyASS sweeps right,
    placing points when the greatest last access time seen increases.
    The staircase represents these increasing last access times.}
  \label{fig:greedyass}
  \caption{(Right) Later in the same execution of \GreedyASS.
  The most recent neighborhoods for
  two of the elements are represented as line segments surrounding
  those elements. Observe that adding
  another search for anything within a neighborhood will result in
  accessing the corresponding element for that neighborhood.}
  \label{fig:neighborhoods}
\end{figure}

\section{An Access Lemma and its Corollaries}
\label{sec:access}

In their paper on splay trees, Sleator and Tarjan prove the
\emph{access lemma}, a very general expression detailing the amortized cost
of a splay (and therefore search) operation~\cite{st-sbst-85}. They
use this lemma to prove several upper bounds, including the
entropy bound, the static finger bound, and the working set bound.
Wang et~al. prove a similar lemma for their multi-splay tree data structure
to show $O(\log \log n)$-competitiveness and $O(\log n)$ amortized performance,
and the version
of the lemma given in Wang's Ph.D. thesis is used to prove
the other distribution sensitive upper bounds listed
above~\cite{cds-lcbst-06,c-mt-06}.
In this section, we provide a similar lemma for
\GreedyASS and discuss its consequences.

\subsection{Potentials and Neighborhoods}

Fix a BST access sequence~$S$ and let~$X = P(S)$. We consider the execution of
\GreedyASS on~$X$.
Let~$\rho(x,i)$ be the last access of~$x$ at or before time~$i$.
Formally,~$\rho(x,i)$ is the~$y$ coordinate of the highest point on the closed
ray from~$(x,i)$ to~$(x,-\infty)$.

Let~$a$ be the greatest positive
integer smaller than~$x$ such that ${\rho(a,i) \geq \rho(x,i)}$
(or let~$a=0$ if no such integer exists).
The \EMPH{left neighborhood} of~$x$ at time~$i$ is
$\Set{a+1,a+2,\dots,x-1}$
and denoted $\Gamma_{\ell}(x,i)$.
The \EMPH{right neighborhood} of~$x$
at time~$i$ is defined
similarly and denoted~$\Gamma_r(x,i)$. Finally, the
\EMPH{inclusive neighborhood} of~$x$ at time~$i$ is
$\Gamma(x,i)=\Gamma_{\ell}(x,i)\cup\Gamma_r(x,i)\cup\Set{x}$.

The inclusive neighborhood of~$x$ at time~$i$
contains precisely those keys whose appearance as~$s_{i+1}$
would prompt \GreedyASS to access~$x$ at time~$i+1$. Intuitively,
the inclusive neighborhood is similar to a node's subtree in the arboral model.
See Fig.~\ref{fig:neighborhoods}.

Assign to each element $x\in\Set{1,\dots,n}$ a positive real
weight~$w(x)$.
The \EMPH{size} of~$x$ at time~$i$ is
$\sigma(x,i)=\sum_{e\in\Gamma(x,i)}w(e)$.
The \EMPH{rank} of~$x$ at time~$i$ is~${r(x,i)=\Floor{\lg \sigma(x,i)}}$.
Finally, define a
potential function ${\Phi(i) = \sum_{x\in[n]} r(x,i)}$ and let
the \EMPH{amortized cost} of a search at time~$i$
be ${1 + |P_i| + \Phi(i)-\Phi(i-1)}$.

\begin{lemma}[Access Lemma]
  \label{lem:access}
  Let~$W=\sum_{x\in[n]}w(x)$.
  The amortized cost of a search at time~$i$ is at most
  ${5+6\Floor{\lg W} - 6r(s_i,i-1)}$.
\end{lemma}

\subsection{Immediate Consequences}

Before we proceed to prove Lemma~\ref{lem:access}, we will show several of its
consequences. Recall that the equivalence between the arboral and geometric
models mean these corollaries apply to both \GreedyASS and \GreedyFuture.
The proofs of these corollaries mirror the proofs by Sleator and
Tarjan for splay trees~\cite{st-sbst-85}.

\begin{corollary}[Balance Theorem]
  \label{cor:balance}
  The total cost of searching is $O((m+n)\times$
  $\log n)$.
\end{corollary}
\begin{corollary}[Static Optimality Theorem]
  Let~$t(x)$ be the number of times~$x$ appears in the search
  sequence~$S$.
  If every element is searched at least once, the total cost of searching is
  $O\Paren{m+\sum_{x=1}^n t(x)\log\Paren{m/t(x)}}$.
\end{corollary}
\begin{corollary}[Static Finger Theorem]
  Fix some element~$f$.
  The total cost of searching is $O(m+\sum_{i=1}^m\log(|s_i-f|+1))$.
\end{corollary}
\begin{corollary}[Working Set Theorem]
  \label{cor:working_set}
  Let~$d(i)$ be the number of distinct elements in the search sequence~$S$
  before~$s_i$ and since the last instance of~$s_i$. If there are no earlier
  instances of~$s_i$, then let~$d(i)=i-1$.
  The total cost of searching is $O(m+\sum_{i=1}^m\log(d(i)+1))$.
\end{corollary}
Note that Corollary~\ref{cor:working_set} implies other upper bounds on
\GreedyFuture's performance such as key-independent
optimality~\cite{i-kio-05}.

\subsection{Telescoping Rank Changes}

We proceed to prove Lemma~\ref{lem:access}. First we observe the following.

\begin{lemma}
  \label{lem:no_access}
  Let~$x$ be any element \emph{not} accessed during search~$i$. Then we have
  $\Gamma(x,i-1)=\Gamma(x,i)$.
\end{lemma}
\begin{proof}
  Assume without loss of generality that~$x>s_i$. Let~$x_{\ell}$ be the greatest
  element in $\Set{s_i,s_i+1,\dots,x-1}$ such that
  $\rho(x_{\ell},i-1)\geq\rho(x,i-1)$.
  Element~$x_{\ell}$ must exist, because \GreedyASS
  does not access~$x$ at time~$i$.
  No elements in $\Set{x_{\ell}+1,\dots,x-1}$ are accessed at time~$i$
  since they have smaller last access time than~$x_{\ell}$,
  so $\Gamma_{\ell}(x,i-1)=\Gamma_{\ell}(x,i)$. Likewise, no elements in~$\Gamma_r(x,i-1)$
  are accessed at time~$i$ since they have smaller last access time than~$x$.
  The inclusive neighborhood of~$x$ (as well as its size and rank) remains
  unchanged by the search.
\end{proof}

Consider a search at time~$i$.
Lemma~\ref{lem:no_access}
immediately implies the amortized cost of the search is equal to
\begin{equation}
  \label{eq:cost_sum}
  \sum_{x \in P_i \cup\Set{s_i}} \Paren{1 + r(x,i)-r(x,i-1)}.
\end{equation}
Suppose we access an element~$x\neq s_i$. Assume~$x > s_i$ without loss of
generality. If it exists,
let~$x_r$ be the least accessed element greater than~$x$. We
call~$x_r$ the \EMPH{successor} of~$x$.
Observe that~$\Gamma(x,i)$ contains
a subset of the elements in~$\Set{s_i+1,\dots,x_r-1}$ while~$\Gamma(x_r,i-1)$
contains a superset of the elements in~$\Set{s_i,\dots,x_r}$. This
fact implies $\Gamma(x,i)\subset\Gamma(x_r,i-1)$ which in turn implies
\begin{equation}
  \label{eq:ranks}
  \sigma(x,i) < \sigma(x_r,i-1) \text{ and } r(x,i) \leq r(x_r,i-1).
\end{equation}
If the second inequality is strict, then
\begin{equation}
  \label{eq:good_telescope}
  1+r(x,i)-r(x,i-1) \leq r(x_r,i-1)-r(x,i-1).
\end{equation}
Otherwise,
\begin{equation}
  \label{eq:bad_telescope}
  1+r(x,i)-r(x,i-1) = 1+r(x_r,i-1)-r(x,i-1).
\end{equation}
Call an accessed element~$x>s_i$
a \EMPH{stubborn element} if~$x$ has a successor~$x_r$
and $r(x,i) = r(x_r,i-1)$.
From (\ref{eq:cost_sum}), (\ref{eq:good_telescope}),
and (\ref{eq:bad_telescope}) above, the amortized cost of accessing elements
greater than~$s_i$ forms a telescoping sum and we derive the following lemma.

\begin{lemma}
  \label{lem:cost_right}
  Let~$\alpha$ be the number of elements greater than~$s_i$ that are stubborn
  and let~$e_{r\ell}$ and~$e_{rr}$ be the least and greatest elements
  greater than~$s_i$ to be accessed. The amortized cost of accessing elements
  greater than~$s_i$ is
  $$1+\alpha+r(e_{rr},i)-r(e_{r\ell},i-1).$$
\end{lemma}

\subsection{Counting Stubborn Elements}
The biggest technical challenge remaining is to upper bound
the number of stubborn elements~$\alpha$. We have
the following lemma.

\begin{lemma}
  \label{lem:stubborn}
  The number of accessed elements greater than~$s_i$ which are stubborn is at
  most
  $$1+2\Floor{\lg W} -2r(s_i,i-1)$$
\end{lemma}
\begin{proof}
  Consider any stubborn element~$x>s_i$ and its successor~$x_r$. Let the
  \EMPH{left size} of~$x$ at time~$i$ be
  $\sigma_{\ell}(x,i)=\sum_{e\in\Gamma_{\ell}(x,i)}w(e)$. Further,
  let the \EMPH{left rank} of~$x$ at time~$i$ be
  $r_{\ell}(x,i)=\Floor{\lg(\sigma_{\ell}(x,i))}$.
  By the definitions of stubborn elements and left sizes we see
  \begin{equation}
    \label{eq:size_gap}
    \sigma(x,i) > \frac{1}{2}\sigma(x_r,i-1) >
      \frac{1}{2}\sigma_{\ell}(x_r,i-1).
  \end{equation}
  We note that for any accessed element~$v$ (stubborn or not) with 
  $s_i < v < x$ we have
  \begin{equation}
    \label{eq:gap_gap}
    \sigma_{\ell}(v,i-1) < \frac{1}{2}\sigma_{\ell}(x_r,i-1)
  \end{equation}
  by (\ref{eq:size_gap})
  since every element of~$\Gamma_{\ell}(v,i-1)$ is
  in~$\Gamma_{\ell}(x_r,i-1)$, but none of these elements are
  in~$\Gamma(x,i)$ since the left neighborhood of~$x$ at time~$i$ cannot
  extend past~$v$. Further,
  \begin{equation}
    \label{eq:s_base}
    \sigma_{\ell}(x,i-1) \geq \sigma(s_i,i-1)
  \end{equation}
  since all weights are positive and every element in~$\Gamma(s_i,i-1)$
  is also in~$\Gamma_{\ell}(x,i-1)$.

  Let~$z>s_i$ be
  the greatest stubborn element, and let~$z_r$ be its successor.
  We will inductively argue the number of stubborn elements is at most
  $$1+2r_{\ell}(z_r,i-1)-2r(s_i,i-1)$$
  which is a stronger statement than that given in the lemma. The argument
  can be divided into two cases.

  \begin{enumerate}
  \item
    Suppose ${\sigma_{\ell}(z_r,i-1)<2\sigma(s_i,i-1)}$.
    For any stubborn element~$v$ between~$s_i$ and~$z$ we have
    $$\sigma_{\ell}(v,i-1) < \sigma(s_i,i-1)$$
    by (\ref{eq:gap_gap}).
    There can be no such element~$v$ by (\ref{eq:s_base}),
    making~$z$ the only stubborn element.
    The total number of stubborn elements is
    \begin{align*}
      1 &\leq 1+2r_{\ell}(z,i-1)-2r(s_i,i-1) \\
        &\leq 1+2r_{\ell}(z_r,i-1)-2r(s_i,i-1)
    \end{align*}
    by (\ref{eq:s_base}) and the definition of left rank.
  \item
    Now suppose $\sigma_{\ell}(z_r,i-1)\geq2\sigma(s_i,i-1)$.
    Consider any stubborn element~$v$
    with successor~$v_r$ such that $s_i < v < v_r < z$.
    Note that if a stubborn element exists with~$z$ as its successor,~$v$
    cannot be this stubborn element.
    We have
    $$\sigma_{\ell}(v_r,i-1) < \frac{1}{2}\sigma_{\ell}(z_r,i-1)$$
    by (\ref{eq:gap_gap}).
    By induction on the left sizes of stubborn element
    successors greater than~$s_i$,
    the successors of at most
    $$1+2\Floor{\lg\Paren{\frac{1}{2}\sigma_{\ell}(z_r,i-1)}}-2r(s_i,i-1)$$
    stubborn elements can have this smaller left size. Counting~$z$
    and the one other stubborn element that may exist with~$z$ as
    its successor, the total number of stubborn
    elements is at most
    $$3+2\Floor{\lg\Paren{\frac{1}{2}\sigma_{\ell}(z_r,i-1)}}-2r(s_i,i-1)
      =1+2r_{\ell}(z_r,i-1)-2r(s_i,i-1).$$
  \end{enumerate}
\end{proof}

\subsection{Finishing the Proof}

We now conclude the proof of Lemma~\ref{lem:access}.
\begin{proof}
  By Lemma~\ref{lem:no_access},
  the amortized cost of accessing~$s_i$ alone is
  $$1+r(s_i,i)-r(s_i,i-1) \leq 5+6\Floor{\lg W}-6r(s_i,i-1)$$
  so the lemma holds in this case.

  If all other accessed elements are greater than~$s_i$, let~$e_{r\ell}$
  and~$e_{rr}$ be the least and greatest of these elements.
  Observe $r(e_{r\ell},i-1) \geq r(s_i,i)$ and $r(e_{rr},i) \leq \Floor{\lg W}$.
  By Lemma~\ref{lem:cost_right} and Lemma~\ref{lem:stubborn},
  the total amortized cost of accessing elements is at most
  \begin{align*}
    \lefteqn{3 + r(s_i,i) - 3r(s_i,i-1) + 2\Floor{\lg W}
      + r(e_{rr},i)-r(e_{r\ell},i-1)} \\
    &\qquad\leq 3 + 3\Floor{\lg W} - 3r(s_i,i-1)\hspace{1.5in} \\
    &\qquad\leq 5 + 6\Floor{\lg W} - 6r(s_i,i-1)
  \end{align*}
  so the lemma holds in this case. It also holds in the symmetric case when
  all accessed elements are smaller than~$s_i$.

  Finally, consider the case when there are accessed elements both greater than
  and less than~$s_i$. Let~$e_{\ell\ell}$
  and~$e_{\ell r}$ be the least and greatest
  elements \emph{less than}~$s_i$. Observe $r(e_{\ell\ell},i)\leq\Floor{\lg W}$
  and $r(e_{\ell r}, i-1)\geq r(s_i,i-1)$. By two applications of
  Lemma~\ref{lem:cost_right} and Lemma~\ref{lem:stubborn}, the total amortized
  cost of the search is at most
  \begin{align*}
    \lefteqn{5+r(s_i,i)-5r(s_i,i-1)+4\Floor{\lg W}+r(e_{rr},i)
      -r(e_{r\ell},i-1)}\\
      &\qquad\quad+r(e_{\ell\ell},i)-r(e_{\ell r},i-1)\hspace{1.5in}\\
    &\qquad\leq 5 +6\Floor{\lg W} - 6r(s_i,i-1)\hspace{1.5in}
  \end{align*}
\end{proof}

\section{A Sequential Access Theorem}

The working set bound proven above shows that \GreedyFuture has good
\emph{temporal locality}. Accessing an element shortly after its last access
guarantees a small amortized search time. Sleator and Tarjan conjectured
that their splay trees also demonstrate good spatial locality properties in the
form of the dynamic finger conjecture~\cite{st-sbst-85}. This conjecture was
verified by Cole, et~al.~\cite{cmss-dfcst1-00,c-dfcst2-00}.

One special case of the dynamic finger theorem considered by Tarjan
and others was the sequential access
theorem~\cite{t-sastl-85,e-satdc-04,cds-lcbst-06,c-mt-06}.
We give a straightforward proof of the sequential
access theorem when applied to \GreedyFuture. Note that this theorem requires
focusing on an arbitrary fixed BST, so we do not use the
geometric model in the proof.

\begin{theorem}[Sequential Access Theorem]
  \label{thm:sequential}
  Let $S=\langle1,2,\dots,n\rangle$. Starting with an arbitrary BST~$T_0$, the
  cost of running \GreedyFuture on search sequence~$S$ is~$O(n)$.
\end{theorem}

Let $T_0, T_1, \dots, T_n$ be the sequence of search trees configured by
\GreedyFuture. We make the following observations:
\begin{lemma}
  \label{lem:left-spine}
  For all~$i>1$, either node~$i$ is the root of~$T_{i-1}$
  or~$i-1$ is the root and~$i$ is the leftmost node of the
  root's right subtree.
\end{lemma}
\begin{proof}
  If~$i$ was accessed during
  the~$i-1$st search, then~$i$ is the root of~$T_{i-1}$. Otherwise,~$i-1$
  is the predecessor node of~$i$
  on the search path. Therefore,~$i-1$ is the root of~$T_{i-1}$
  and~$i$ is the leftmost node of the root's right subtree.
\end{proof}

\begin{lemma}
  \label{lem:deep_access}
  Node~$x$ is accessed at most once in any position other than the root
  or the root's right child.
\end{lemma}
\begin{proof}
  Consider node~$x$ and search~$i$. Node~$x$ cannot be accessed if~$x<i-1$
  according to Lemma~\ref{lem:left-spine}.
  If~$x$ lies on the search path and
  $x \leq i+1$
  then either~$x$ becomes the root or~$x$ moves into the
  root's left subtree so that~$i$ or~$i+1$
  can become the root.

  Now suppose~$x$ lies along the search path and $x>i+1$. Let~$x_{\ell}$ be the
  least node strictly smaller than~$x$ that does not become the root.
  If~$x_{\ell}$ does not exist, then~$x$ becomes the root's right child
  as either~$x$ is the successor of~$i+1$ on the search path, node~$i+1$ is
  on the search path and~$x=i+2$, or node~$i+1$ is on the search path
  and~$x$ is the successor of~$i+2$ on the search path.
  If~$x_{\ell}$ does exist, then~$x_{\ell}$ becomes the root's right child
  for one of the reasons listed above
  and~$x$ becomes a right descendent of~$x_{\ell}$.

  Node~$x$ cannot be moved to the left subtree of the root's right child
  in all the cases above.
  Lemma~\ref{lem:left-spine} therefore implies~$x$
  is accessed in the root's left subtree on the first search,~$x$
  is accessed \emph{once} in the left subtree of the root's right child, or~$x$
  is never accessed anywhere other than as the root or root's right child.
\end{proof}

We now conclude the proof of Theorem~\ref{thm:sequential}.
\begin{proof}
  The cost of the first search is at most~$n$. The costs of all subsequent
  searches
  is at most~$2(n-1) + n$ according to Lemma~\ref{lem:deep_access};
  at most~$2(n-1)$ node accesses occur at the root or root's right
  child, and at most~$n$ nodes are accessed exactly once in a position other
  than the root or the root's right child.
  The total cost of all searches is at most~$4n-2$.
\end{proof}

\section{Closing Remarks}

The ultimate
goal of this line of research is to prove \GreedyFuture or splay trees
optimal, but showing other upper bounds may prove interesting. In
particular, it would be interesting to see if some difficult to prove
splay tree properties such as the dynamic finger bound have concise proofs when
applied to \GreedyFuture. Another direction is to explore how
\GreedyFuture may be modified to support insertions and deletions while
still maintaining its small search cost.\\

\noindent
\textbf{Acknowledgements}
The author would like to thank Alina Ene, Jeff Erickson, Benjamin Moseley,
and Benjamin Raichel for their advice and helpful discussions as well as
the anonymous reviewers for their suggestions on improving this report.

This research is supported in part by the Department of Energy Office of Science
Graduate Fellowship Program (DOE SCGF), made possible in part by the
American Recovery and Reinvestment Act of 2009, administered by ORISE-ORAU under
contract no. DE-AC05-06OR23100.

\bibliographystyle{splncs}
\bibliography{greedyfuture}

\begin{thebibliography}{10}
\providecommand{\url}[1]{\texttt{#1}}
\providecommand{\urlprefix}{URL }

\bibitem{bddf-lcbst-10}
Bose, P., Dou\"{i}eb, K., Dujmovi\'{c}, V., Fagerberg, R.: An {$O(\log \log
  n)$}-competitive binary search tree with optimal worst-case access times. In:
  Proc. 12th Scandinavian Symp. and Workshop on Algorithm Theory. pp. 38--49
  (2010)

\bibitem{c-dfcst2-00}
Cole, R.: On the dynamic finger conjecture for splay trees. {Part II}: The
  proof. SIAM J. Comput.  30,  44--85 (2000)

\bibitem{cmss-dfcst1-00}
Cole, R., Mishra, B., Schmidt, J., Siegel, A.: On the dynamic finger conjecture
  for splay trees. {Part I}: Splay sorting $\log n$-block sequences. SIAM J.
  Comput.  30,  1--43 (2000)

\bibitem{dhikp-gbst-09}
Demaine, E.D., Harmon, D., Iacono, J., Kane, D., P\u{a}tra\c{s}cu, M.: The
  geometry of binary search trees. In: Proc. 20th ACM/SIAM Symposium on
  Discrete Algorithms. pp. 496--505 (2009)

\bibitem{dhip-doa-07}
Demaine, E.D., Harmon, D., Iacono, J., P\u{a}tra\c{s}cu, M.: Dynamic
  optimality--almost. SIAM J. Comput.  37(1),  240--251 (2007)

\bibitem{dsw-lbfbs-05}
Derryberry, J., Sleator, D.D., Wang, C.C.: A lower bound framework for binary
  search trees with rotations. Tech. Rep. CMU-CS-05-187, Carnegie Mellon
  University (2005)

\bibitem{e-satdc-04}
Elmasry, A.: On the sequential access theorem and deque conjecture for splay
  trees. Theoretical Computer Science  314(3),  459--466 (2004)

\bibitem{gg-dobst-11}
Goyal, N., Gupta, M.: On dynamic optimality for binary search trees.
  http://arxiv.org/abs/1102.4523 (2011)

\bibitem{i-kio-05}
Iacono, J.: Key independent optimality. Algorithmica  42,  3--10 (2005)

\bibitem{l-cfcbs-88}
Lucas, J.M.: Canonical forms for competitive binary search tree algorithms.
  Tech. Rep. DCS-TR-250, Rutgers University (1988)

\bibitem{m-ocls-00}
Munro, J.I.: On the competitiveness of linear search. In: Proc. 8th Annual
  European Symposium on Algorithms. pp. 338--345 (2000)

\bibitem{st-sbst-85}
Sleator, D.D., Tarjan, R.E.: Self-adjusting binary search trees. Journal of the
  Association for Computing Machinery  32(3),  652--686 (1985)

\bibitem{t-sastl-85}
Tarjan, R.E.: Sequential access in splay trees takes linear time. Combinatorica
   5,  367--378 (1985)

\bibitem{c-mt-06}
Wang, C.C.: Multi-Splay Trees. Ph.D. thesis, Carnegie Mellon University (2006)

\bibitem{cds-lcbst-06}
Wang, C.C., Derryberry, J., Sleator, D.D.: {$O(\log \log n)$}-competitive
  binary search trees. In: Proc. 17th Ann. ACM-SIAM Symp. Discrete Algorithms.
  pp. 374--383 (2006)

\bibitem{w-lbabs-89}
Wilber, R.E.: Lower bounds for accessing binary search trees with rotations.
  SIAM J. Comput.  18(1),  56--67 (1989)

\end{thebibliography}
\end{document}